\documentclass[12pt, a4paper,reqno]{amsart}

\usepackage[T1]{fontenc}
\usepackage[utf8]{inputenc}

\usepackage{color}
\usepackage{xcolor}
\definecolor{bleu_sombre}{rgb}{0,0,0.6}  \definecolor{rouge_sombre}{rgb}{0.8,0,0}\definecolor{vert_sombre}{rgb}{0,0.6,0}
\usepackage[plainpages=false,colorlinks,linkcolor=bleu_sombre,
citecolor=rouge_sombre,urlcolor=vert_sombre,breaklinks]{hyperref}

\usepackage[english]{babel}
\usepackage{geometry}

\usepackage{amsmath,amssymb,amsthm,graphicx,amsfonts,url,enumerate,dsfont,stmaryrd, mathrsfs, csquotes, harpoon}

\theoremstyle{plain}
\newtheorem{theorem}{{Theorem}}[section]
\newtheorem*{theorem*}{{Theorem}}
\newtheorem{proposition}[theorem]{Proposition}

\newtheorem*{proposition*}{Proposition}

\newtheorem*{corollary*}{Corollary}
\newtheorem{lemma}[theorem]{Lemma}

\newtheorem*{lemma*}{Lemma}
\theoremstyle{definition}

\newtheorem*{definition*}{Definition}
\theoremstyle{remark}
\newtheorem{remark}[theorem]{Remark}

\makeatletter

\@addtoreset{equation}{section}
\makeatother

\renewcommand{\leq}{\leqslant}	\renewcommand{\geq}{\geqslant}

\newcommand{\R}{\mathbb{R}}

\newcommand{\dd}{\mathrm{d}}

\renewcommand{\Re}{\mathrm{Re}\,}
\renewcommand{\Im}{\mathrm{Im}\,}

\begin{document}

\title[]{On Duclos-Exner's conjecture about waveguides in strong uniform magnetic fields}
\author[E. Bon-Lavigne]{Enguerrand Bon-Lavigne}

\author[L. Le Treust]{Loïc Le Treust}

\author[N. Raymond]{Nicolas Raymond}

\author[J. Royer]{Julien Royer}

\begin{abstract}
We consider the Dirichlet Laplacian with uniform magnetic field on a curved strip in two dimensions. We give a sufficient condition ensuring the existence of the discrete spectrum in the strong magnetic field limit.

\end{abstract}

\maketitle

\section{Introduction and statement of the main results}

\newcommand{\blue}{\color{bleu_sombre}}

In this article, we address the question of existence of the discrete spectrum for a magnetic Laplacian with Dirichlet boundary condition on a two-dimensional curved waveguide.

\subsection{What is a waveguide?}
Let $\gamma : \R \to \R^2$ be a smooth and injective curve with $|\gamma'|= 1$. 
We set $\mathbf N = (\gamma')^\bot$, where for $(a,b) \in \R^2$ we write $(a,b)^\bot$ for $(-b,a)$. 
We denote by $\kappa$ the algebraic curvature of $\gamma$. It is defined by
\[\gamma'' = \kappa \mathbf N\,.\] 
In this article, we work under the assumption that $\kappa$ is compactly supported. For $\delta > 0$ small enough, the function 
\[
\Theta : \left\{ \begin{array}{ccc} \R \times (-\delta,\delta) & \to & \R^2 \\ (s,t) & \mapsto & \gamma(s) + t \mathbf N(s) \end{array} \right.
\]
is injective. We set 
\[\Omega = \Omega_{\gamma,\delta} = \Theta (\Omega_0)\,,\quad \mbox{ with }\quad\Omega_0 = \Omega_{0,\delta} =\R \times (-\delta,\delta)\,.\]
The open set $\Omega$ is what we call a waveguide in this work. %For further use, we define the two boundaries of our waveguide $\Gamma_\pm = \Theta (\R \times \{ \pm \delta \})$, so that $\partial \Omega = \Gamma_+ \sqcup \Gamma_-$.

\subsection{The magnetic Laplacian with Dirichlet boundary conditions}\label{sec.ML}
The waveguide $\Omega$ is subject to a perpendicular uniform magnetic field with intensity $B$. That is why we consider a vector potential $\mathbf A = (A_1,A_2)$ that is smooth on $\overline{\Omega}$, and such that
\begin{equation}\label{eq.rotA}
\partial_{x_1} A_2 - \partial_{x_2} A_1 = 1\,.
\end{equation}
A fundamental property related to magnetic problems on simply connected domains is the gauge invariance. It is nothing but the fact that \eqref{eq.rotA} only defines $\mathbf{A}$ up to adding a gradient vector field. Of course, it is trivial that there is a smooth solution to \eqref{eq.rotA}, since it is sufficient to consider $\mathbf{A}=(0,x_1)$. Actually, one will see that there is a natural choice of vector potential in our setting. Finding a gauge that is adapted to the structure of the waveguide is in fact part of our problem and it has been tackled in the past; see, for instance, \cite{E93} where a curvature-dependent gauge is introduced. We now assume that $\mathbf{A}$ can be chosen smooth on $\overline{\Omega}$ and bounded with bounded derivatives (at any order). It will be explained in Proposition \ref{prop.phi} that we may indeed assume this. 

For $B > 0$, we consider on $\Omega$ the magnetic Laplacian corresponding to the uniform field equal to $B$:
\begin{equation}\label{eq.B}
(-i\nabla - B \mathbf A)^2 - B\,,
\end{equation}
subject to Dirichlet boundary conditions. The subtraction of $B$ is made for the convenience of the analysis and does not change the presence or absence of discrete spectrum (it is based on relating
%interpreting 
the Schr\"odinger operator 
%as 
to the square of a Dirac operator). In order to use semiclassical analysis we also introduce the positive parameter $h = B^{-1}$ and set 
\[
\mathscr{P}_{h}=(-ih\nabla-\mathbf{A})^2-h.
\]
The operator $\mathscr{P}_{h}$ is well defined and selfadjoint on the domain 
\[
\mathsf{Dom}(\mathscr{P}_{h}) = H^1_0(\Omega) \cap H^2(\Omega)\,.
\]

\subsection{A subtle question and a conjecture by P. Duclos and P. Exner}
Our aim is to study the existence of the discrete spectrum of $\mathscr{P}_h$ in the semiclassical limit $h\to 0$ (equivalent to the large magnetic field limit, see \eqref{eq.B}). This question of existence is actually subtle since, when $h$ goes to $0$, not only the bottom of the spectrum moves, but also the bottom of the essential spectrum. In this limit, it is natural to wonder if the bottom of the spectrum stays away from the threshold of the essential spectrum or collides with it. This question is all the more appealing that, when the magnetic field is zero, that is when considering the Dirichlet Laplacian on a strip, one knows that the discrete spectrum always exists as soon as the strip is not straight (see, for instance, \cite{DE95} or the book \cite[Chapter 1]{EK15}). It is also known that (variable) magnetic fields can play against the existence of the discrete spectrum. Such considerations can be found in \cite[Theorem 2.8 \& Proposition 2.11]{KR14} where a magnetic Hardy inequality is proved when the magnetic field has \emph{compact support} and used to establish that the discrete spectrum is empty when the magnetic field is strong enough (see also the original work \cite{EK05})\footnote{Let us also mention that, in \cite{KR14}, the spectrum is also analyzed (by means of resolvent convergence) in the shrinking limit $\delta\to0$ with a possibly $\delta$-dependent magnetic field. Deriving effective operators in such regimes can actually be done in a quite general framework, see \cite{HLT19}. }. 

In the mid nineties, buoyed by the momentum of their work \cite{DE95}, Pierre Duclos and Pavel Exner conjectured that \emph{the discrete spectrum of \eqref{eq.B} is empty when the magnetic field is strong enough (\underline{and uniform})}. This conjecture was explicitely formulated ten years ago during an "Open Problems" session in Barcelona, see \cite{Exner12}. Our main result disproves the conjecture when the waveguide has a \emph{fixed} width $\delta$ assumed to be small enough, but \emph{independently} of $B$.

\subsection{Main result}
Our main result is the following.

\begin{theorem} \label{th.main}
Assume that $\kappa^2$ has a unique maximum, which is non-degenerate. There exist $\delta_0 > 0$ and $h_0 > 0$ such that for all $\delta \in (0,\delta_0)$ and all $h \in (0,h_0)$ we have 
\[
\inf \mathrm{sp}(\mathscr{P}_{h}) < \inf \mathrm{sp}_{\mathsf{ess}} (\mathscr{P}_{h}).
\]
In particular, $\mathscr{P}_{h}$ has non-empty discrete spectrum.
\end{theorem}

We can be more precise and provide some bounds for the bottoms of spectrum and essential spectrum. For this we compare the spectral properties of the magnetic Laplacian on $\Omega$ to those on $\Omega_0$. On $\Omega_0$ we set $\mathbf A_0(s,t) = (-t,0)$ and we consider on $L^2(\Omega_0)$ the operator $\mathscr{P}_{h,0} = (-ih\nabla - \mathbf A_0)^2 - h$, with Dirichlet boundary conditions. Then, we have the following result about the essential spectrum of $\mathscr{P}_{h}$.

\begin{proposition}\label{prop.ess}
For $h > 0$ we set
\[
\lambda_{\mathrm{ess}}(h) = \inf \mathrm{sp}(\mathscr{P}_{h,0}).
\]
Then
\[
\mathrm{sp}_{\mathrm{ess}}(\mathscr{P}_h)=\mathrm{sp}_{\mathrm{ess}}(\mathscr{P}_{h,0})=\mathrm{sp}(\mathscr{P}_{h,0})=[\lambda_{\mathrm{ess}}(h),+\infty)
\]
and
\[
\lambda_{\mathrm{ess}}(h)\geq \frac{(\pi h)^2}{4\delta^2} e^{-\delta^2/h}\,. 
\]
\end{proposition}

To prove an upper bound on the bottom of the spectrum we first introduce on $\Omega_0$ the function $\phi_0$ defined by 
\[
\phi_0(s,t)=\frac{t^2-\delta^2}2.
\]
Then we define $\hat \phi_0 = \phi_0 \circ \Theta^{-1} \in \mathscr C^\infty(\overline \Omega)$. In particular, $\hat \phi_0$ vanishes on $\partial \Omega$. In order to perform the analysis of the bottom of the spectrum, we will use a function $\phi$, looking like $\hat \phi_0$ at infinity, defined thanks to the following proposition.

\begin{proposition}\label{prop.phi}
There exists a unique $\phi\in\mathscr{C}^\infty(\overline{\Omega})$ such that $\Delta\phi=1$, $\phi_{|\partial\Omega}=0$, and $\phi-\hat \phi_0\in\mathscr{S}(\overline{\Omega})$. Moreover, there exists $c_0>0$ such that $\partial_\nu \phi \geq c_0$ on $\partial\Omega$, $\nu$ being the outward pointing normal to the boundary.
\end{proposition}

Then, by gauge invariance, we can choose $\mathbf A = \nabla \phi^\bot$ in the definition of $\mathscr P_h$. In particular, we may assume that $\mathbf A$ is bounded on $\Omega$, as announced in Section \ref{sec.ML}. Here comes our result ensuring the existence of the discrete spectrum.

\begin{theorem}\label{thm.main}
Assume that $\phi$ given by Proposition \ref{prop.phi} has a unique minimum $\phi_{\min}$ (reached at $x_{\min} \in \Omega$), which is non-degenerate and smaller than $\min \phi_0 = - \delta^2 / 2$. Then, as $h \to 0$ we have
\[\inf\mathrm{sp}(\mathscr{P}_h)\leq \frac{J}{\pi}\sqrt{\det\mathrm{Hess}_{x_{\min}}\phi}e^{2\phi_{\min}/h} \big(1+o(1)\big)\,,\]
with
\[J=\underset{f\in\mathscr{E}}{\inf}\|(\partial_\nu \phi) ^{\frac12}f\|^2_{\partial\Omega}\,,\]
and
\[\mathscr{E}=\{f\in\mathscr{O}(\Omega)\cap H^1(\Omega) :  f(x_{\min})=1\}\,,\]
where $\mathscr{O}(\Omega)$ is the set of holomorphic functions on $\Omega$.
\end{theorem}

\begin{remark}
	~
\begin{enumerate}[\rm i)]
\item The set $\mathscr{E}$ is not empty as we can see by considering a function of the form $f\colon z\mapsto c(z-z_1)^{-2}$ with $z_1\notin\overline{\Omega}$ and $c$ such that $f(x_{\min})=1$.
\item Due to a classical trace theorem and the fact that $\partial_\nu \phi$ is bounded, $J$ is finite.
\item The fact that $\phi$ has a unique minimum (which is non degenerate) can be ensured under explicit assumptions on the curvature $\kappa$ and on the width of the waveguide, see Proposition \ref{prop.assumptionvalid} below. 
\item By using Proposition \ref{prop.ess} and under the assumption on $\phi$ in Theorem \ref{thm.main}, we have $\inf \mathrm{sp}(\mathscr P_h) < \inf \mathrm{sp}_{\mathsf{ess}}(\mathscr P_h)$.
\end{enumerate}
\end{remark}

Our proof of Theorem \ref{thm.main} is based on extensions of strategies used in \cite{BLTRS21}\footnote{motivated by the seminal works \cite{EKP16} and \cite{HP17}.}, where the asymptotic simplicity of the low-lying eigenvalues is established, under generic assumptions on $\Omega$. Let us emphasize that, in \cite{BLTRS21}, $\Omega$ is assumed to be \emph{bounded} and that the assumption on $\phi$ can be ensured, in the uniform magnetic field case, when $\Omega$ is \emph{strictly convex} (thanks to the works by Kawohl \cite{K85, K85b}). In the present setting, $\Omega$ is neither bounded, nor convex. Moreover, in our unbounded setting, one needs to be very careful since the functional spaces (such as the Hardy spaces) involved in \cite{BLTRS21} are no more obviously well-defined. The study of such spaces on strips\footnote{which started a long time ago, see, for instance, \cite{W42}.} has an interest of its own and their use to deduce precise spectral asymptotics will be the object of a future work. Fortunately, we do not need them to disprove Duclos-Exner's conjecture.

To complete our analysis, it remains to give a sufficient condition under which the assumption of Theorem \ref{thm.main} is satisfied.

\begin{proposition}\label{prop.assumptionvalid}
Assume that $\kappa\in\mathscr{C}^\infty_0(\R)$ and that $\kappa^2$ has a unique maximun, which is non-degenerate. There exists $\delta_0>0$ such that, for all $\delta\in(0,\delta_0)$, $\phi$ has a unique minimum in $\Omega$, which is non-degenerate. Moreover, $\phi_{\min}<(\phi_0)_{\min}$.	
\end{proposition}

Theorem \ref{th.main} follows from Proposition \ref{prop.ess}, Proposition \ref{prop.phi}, Theorem \ref{thm.main} and Proposition \ref{prop.assumptionvalid}. Due to our motivation to disprove a conjecture from the nineties, we provide the reader with rather self-contained proofs (and sometimes recall basic arguments). In Section \ref{sec.2}, we analyze the essential spectrum and we prove Proposition \ref{prop.ess}. In Section \ref{sec.phi}, the existence of the function $\phi$ is established and we prove Propositions \ref{prop.phi} and \ref{prop.assumptionvalid}. In Section \ref{sec.proof}, we prove Theorem \ref{thm.main}.

\section{The Essential Spectrum}\label{sec.2}

In this section, we prove Proposition \ref{prop.ess}, which follows from Lemmas \ref{prop.Ph0} and \ref{lem.Weyl}.
We first recall a classical result.
\begin{lemma}\label{lem.ephih}
Let $\phi \in \mathcal C^\infty(\overline \Omega)$ be bounded with bounded derivatives and $A = \nabla \phi^\bot$. For all $\psi\in H^1_0(\Omega)$, we have
\[\|(-ih\nabla-\mathbf{A})\psi\|_{L^2(\Omega)}^2-h\|\psi\|_{L^2(\Omega)}^2=4h^2\int_{\Omega}e^{-2\phi/h}|\partial_{\overline{z}} u|^2\mathrm{d}x\,,\]
where $u:=e^{\phi/h}\psi\in H^1_0(\Omega)$.
\end{lemma}
\begin{proof}
We have 
\[\begin{split}
4h^2\int_{\Omega}e^{-2\phi/h}|\partial_{\overline{z}} u|^2\mathrm{d}x&=\int_{\Omega}|e^{-\phi/h}(h\partial_{1}+ih\partial_2) u|^2\mathrm{d}x\\
&=\int_{\Omega}|(h\partial_{1}+ih\partial_2) e^{-\phi/h}u-[h\partial_1+i\partial_2,e^{-\phi/h}]u|^2\mathrm{d}x\\
&=\int_{\Omega}|(h\partial_{1}+i\partial_2\phi+ih\partial_2+\partial_1\phi) \psi|^2\mathrm{d}x\\
&=\int_{\Omega}|(h\partial_{1}-iA_1+ih\partial_2+A_2) \psi|^2\mathrm{d}x\\
&=\int_{\Omega}|(L_1+iL_2) \psi|^2\mathrm{d}x\,,\quad L_j=-ih\partial_j-A_j\,.\\
\end{split}\]
Then, we get
\[
\begin{split}
4h^2\int_{\Omega}e^{-2\phi/h}|\partial_{\overline{z}} u|^2\mathrm{d}x&=\|(-ih\nabla-\mathbf{A})\psi\|^2+2\Re\langle L_1\psi,iL_2\psi\rangle\\
&=\|(-ih\nabla-\mathbf{A})\psi\|^2+2\Im\langle L_1\psi,L_2\psi\rangle.\\
\end{split}
\]
Note that	
\begin{align*}
2\Im\langle L_1\psi,L_2\psi\rangle
& =2\Im\langle \psi,L_1L_2\psi\rangle\\
& =2\Im\langle \psi,L_2L_1\psi+[L_1,L_2]\psi\rangle\\
& =2\Im\langle L_2\psi,L_1\psi\rangle-2h\,.
\end{align*}
The conclusion follows.
\end{proof}

\begin{proposition}\label{prop.Ph0}
For all $h>0$ we have
\[\mathrm{sp}(\mathscr{P}_{h,0})=[\lambda_{\mathrm{ess}}(h),+\infty)\,,\]
and
\[\lambda_{\mathrm{ess}}(h)\geq \frac{(\pi h)^2}{4\delta^2} e^{-\delta^2/h}\,. \]		
\end{proposition}
\begin{proof}
By using the Fourier transform, we have
\[\mathscr{P}_{h,0}=\int^\oplus \mathscr{P}_{h,0,\xi}\dd\xi\,,\]
where the operator
\[\mathscr{P}_{h,0,\xi}=-h^2\partial_t^2+(\xi+t)^2-h\]
is equipped with the Dirichlet conditions at $t=\pm\delta$. Let us denote by $(\gamma_n(\xi,h))_{n\geq 1}$ the increasing sequence of its eigenvalues. A straightforward application of the min-max theorem shows that, for all $h>0$,
\[\lim_{\xi\to\pm\infty}\gamma_n(\xi,h)=+\infty\,.\]
We get
\[\mathrm{sp}(\mathscr{P}_{h,0})=[\min_{\xi\in\R}\gamma_1(\xi,h),+\infty)=\mathrm{sp}_{\mathrm{ess}}(\mathscr{P}_{h,0})\,.\]
By the min-max principle, we have
\[\inf \mathrm{sp}(\mathscr{P}_{h,0})=\inf_{\psi\in H^1_0(\Omega_0)\setminus\{0\}}\frac{\|(-ih\nabla-\mathbf{A}_0)\psi\|^2-h\|\psi\|^2}{\|\psi\|^2}\,,\]
and, by letting $\psi=e^{-\phi_0/h}u$, we get
\[\inf \mathrm{sp}(\mathscr{P}_{h,0})=\inf_{u\in H^1_0(\Omega_0)\setminus\{0\}}\frac{4h^2\|e^{-\phi_0/h}\partial_{\overline{z}}u\|^2}{\|e^{-\phi_0/h}u\|^2}\,.\]
This allows to get the rough lower bound
\begin{align*}
\inf \mathrm{sp}(\mathscr{P}_{h,0})
& \geq e^{-\delta^2/h}\inf_{u\in H^1_0(\Omega_0)\setminus\{0\}}\frac{4h^2\|\partial_{\overline{z}}u\|^2}{\|u\|^2}\\
& \geq h^2e^{-\delta^2/h}\lambda^{\mathrm{Dir}}_1((-\delta,\delta))\\
& \geq \frac{(\pi h)^2}{4\delta^2} e^{-\delta^2/h} \,.
\end{align*}
This last argument already appeared in \cite[Theorem 3.1]{HP17}.
\end{proof}

Let us recall the following classical result.
\begin{lemma}\label{lem.Weylgen}
Consider $(T_1,\mathrm{Dom}(T_1))$ and $(T_2,\mathrm{Dom}(T_2))$ two closed operators on a Banach space $E$. Assume that there exists $z_0\in\rho(T_1)\cap\rho(T_2)$ such that the operator 
$K : (T_1-z_0)^{-1}-(T_2-z_0)^{-1} : E\to E$ is compact. Then,
\[\mathrm{sp}_{\mathrm{ess}}(T_1)=\mathrm{sp}_{\mathrm{ess}}(T_2)\,.\]
\end{lemma}
\begin{proof}
Let us recall the proof and note that it does not require the selfadjointness of $T_1$ or $T_2$. We recall that $\lambda\in\mathrm{sp}_{\mathrm{ess}}(T_1)$ if and only if $T_1-\lambda$ is not a Fredholm operator with index $0$.	

Consider $\lambda\notin\mathrm{sp}_{\mathrm{ess}}(T_1)$ and write
\[\begin{split}
T_2-\lambda=T_2-z_0+(z_0-\lambda)&=\left(\mathrm{Id}+(z_0-\lambda)(T_2-z_0)^{-1}\right)(T_2-z_0)\\
&=\left(\mathrm{Id}+(\lambda-z_0)K+(z_0-\lambda)(T_1-z_0)^{-1}\right)(T_2-z_0)\\
&=\left((\lambda-z_0)K+(T_1-\lambda)(T_1-z_0)^{-1}\right)(T_2-z_0)\,.
\end{split}
\]
Now, notice that $T_2-z_0 : \mathrm{Dom}(T_2)\to E$ is Fredholm with index $0$ (since it is bijective). The operator $(T_1-z_0)^{-1} : E\to \mathrm{Dom}(T_1)$ is also bijective and thus Fredholm with index $0$. Therefore $(T_1-\lambda)(T_1-z_0)^{-1} : E\to E$ is also Fredholm with index $0$ (see \cite[Corollary 5.7]{CR21}). Since $K$ is compact, 
\[(\lambda-z_0)K+(T_1-\lambda)(T_1-z_0)^{-1}\] 
is still Fredholm with index $0$ (see \cite[Corollary 5.9]{CR21}). Thus, $T_2-\lambda$ is Fredholm with index $0$ (again by \cite[Corollary 5.7]{CR21}).
\end{proof}

Thanks to Lemma \ref{lem.Weylgen}, it is rather easy to get the following.

\begin{lemma}\label{lem.Weyl}
For all $h>0$, we have $\mathrm{sp}_{\mathrm{ess}}(\mathscr{P}_h)=\mathrm{sp}_{\mathrm{ess}}(\mathscr{P}_{h,0})$.	
\end{lemma}
\begin{proof}
The operator $\mathscr{P}_h$ is unitarily equivalent to the selfadjoint operator $\widetilde{\mathscr{P}}_h$	(on $L^2(\Omega_0,\mathrm{d}s\mathrm {d}t)$) given by
\[\widetilde{\mathscr{P}}_h=-\partial^2_t+(a^{-\frac12}(D_s-\tilde A(s,t)) a^{-\frac12})^2-\frac{\kappa^2}{4a^2}-h\,,\quad a(s,t)=1-t\kappa(s)\,,\]
where $\tilde A(s,t)=t-\kappa(s)\frac{t^2}{2}$.
Since $\kappa$ is compactly supported, we see that $\widetilde{\mathscr{P}}_h$ acts as $\mathscr{P}_{h,0}$ away from a compact set.

Let us now apply Lemma \ref{lem.Weylgen} with $T_1=\mathscr{P}_{h,0}$, $T_2=\widetilde{\mathscr{P}}_h$ and $z_0=i$. The resolvent formula gives
\[K=(T_1-z_0)^{-1}(T_2-T_1)(T_2-z_0)^{-1}\,.\]
In our case, we have
\[T_2-T_1=a^{-\frac12}\left[(D_s-\tilde A) a^{-1}(D_s-\tilde A)\right]a^{-\frac12}-(D_s-t)^2-\frac{\kappa^2}{4a^2}\,.\]
Computing some commutators shows that we can find three smooth functions on $\overline{\Omega_0}$, compactly supported with respect to $s$, $W_1$, $W_2$ and $W_3$ such that
\[T_2-T_1=W_1(s,t)D^2_s+W_2(s,t)D_s+W_3(s,t)\,.\]
Then, by elliptic regularity and the Kolmogorov-Riesz theorem (see \cite[Theorem 4.14 \& Remark 4.15]{CR21}), we notice that $W(\widetilde{\mathscr{P}}_h-i)^{-1} : L^2(\Omega_0)\to H^1(\Omega_0)$ is compact for all $W\in\mathscr{C}^\infty_0(\overline{\Omega_0})$. This shows that the terms involving $W_2$ and $W_3$ in $K$ are compact operators on $L^2(\Omega_0)$ (by using that the set of compact operators forms an ideal). Concerning the term involving $W_1$, we notice, on the one hand, that $D^2_s(\widetilde{\mathscr{P}}_h-i)^{-1}$ is bounded on $L^2(\Omega_0)$ and, on the other hand, that $(\mathscr{P}_{h,0}-i)^{-1}W_1 : L^2(\Omega_0)\to L^2(\Omega_0)$ is compact since the operators
\[[(\mathscr{P}_{h,0}-i)^{-1},W_1]=-(\mathscr{P}_{h,0}-i)^{-1}[\mathscr{P}_{h,0},W_1](\mathscr{P}_{h,0}-i)^{-1}\]
and $W_1(\mathscr{P}_{h,0}-i)^{-1} :  L^2(\Omega_0)\to L^2(\Omega_0)$ are compact.

Applying Lemma \ref{lem.Weylgen}, the conclusion follows.
\end{proof}

\section{On the function \texorpdfstring{$\phi$}{phi}}\label{sec.phi}

In this section we prove Propositions \ref{prop.phi} and \ref{prop.assumptionvalid}. We recall that $\phi_0$ and $\tilde \phi_0$ were defined before Proposition \ref{prop.phi}.

\subsection{Proof of Proposition \ref{prop.phi}}

Assume that two functions $\phi_1$ and $\phi_2$ satisfy the conclusions of the proposition. Then $\phi_1 - \phi_2$ is harmonic in $\Omega$ and belongs to $H^1_0(\Omega)$.
This implies that $\phi_1 = \phi_2$ and gives uniqueness.

 Since the tube $\Omega$ is straight at infinity, we have $\Delta\hat \phi_0=1$ outside a compact set. In particular, $1-\Delta\hat \phi_0\in L^2(\Omega)$. By the Poincar\'e inequality (see, for instance, \cite{DE95} for the case of a waveguide) and the Riesz representation theorem, there exists a unique $f_0\in H^1_0(\Omega)$ such that 
 \[
 \forall \varphi\in H^1_0(\Omega), \quad \int_{\Omega}\nabla f_0\cdot\nabla\varphi\,\mathrm{d}x=\int_\Omega (1-\Delta\hat \phi_0) \varphi\,\mathrm{d}x
 \]
Then $-\Delta f_0 =1-\Delta\hat \phi_0$ in the sense of distributions, and $f_0$ belongs to $\mathscr C^\infty(\overline\Omega)$ by elliptic regularity.

Let $V=1-\Delta\hat \phi_0$ and consider a non-negative and bounded Lipschitzian function $\Phi$ on $\Omega$. We have
\[\langle -\Delta f_0,e^{2\Phi}f_0\rangle=\int_{\Omega}Ve^{2\Phi}f_0\mathrm{d}x\,.\]
Taking the real part and integrating by parts in the left-hand-side, we get the "Agmon formula"
\[
\|\nabla(e^{\Phi}f_0)\|_{L^2(\Omega)}^2-\|f_0e^{\Phi}\nabla\Phi\|_{L^2(\Omega)}^2=\Re\int_{\Omega}Ve^{2\Phi}f_0\mathrm{d}x\,.
\]
Since $V$ has compact support, it follows that
\[
\|\nabla(e^{\Phi}f_0)\|_{L^2(\Omega)}^2-\|f_0e^{\Phi}\nabla\Phi\|_{L^2(\Omega)}^2\leq \|V\|_{L^2(\Omega)}\|f_0\|_{L^2(\Omega)}\underset{\mathrm{supp V}}{\max} e^{2\Phi}\,.
\]
By the Poincar\'e inequality we have
\[
\|\nabla(e^{\Phi}f_0)\|^2\geq\lambda_1(\Omega)\|e^{\Phi}f_0\|^2\,,
\]
where $\lambda_1(\Omega)>0$ is the infimum of the spectrum of the Dirichlet Laplacian on $\Omega$. This shows that
\[
\big( \lambda_1(\Omega)-\|\nabla\Phi\|^2_\infty \big) \int_{\Omega}e^{2\Phi}|f_0|^2\mathrm{d}x\leq  \|V\|_{L^2(\Omega)}\|f_0\|_{L^2(\Omega)}\underset{\mathrm{supp V}}{\max} e^{2\Phi}\,.
\]
Choosing $\Phi(x)=\Phi_m(x)=\alpha\min(\langle x\rangle,m)$ (with $\alpha>0$ fixed small enough) and letting $m\to+\infty$, we see by the Fatou lemma that there exists $C > 0$ such that
\[
\int_{\Omega}e^{2\alpha\langle x\rangle}|f_0|^2\mathrm{d}x \leq C \|f_0\|_{L^2(\Omega)}\,.
\]
Coming back to the Agmon formula, we also see that $f_0$ exponentially decays in $H^1$-norm. By means of elliptic estimates, we can check that it is also the case in $H^{k}(\Omega)$ for all $k$. This proves in particular that $f_0$ belongs to the Schwartz class $\mathscr{S}(\overline{\Omega})$.

We set $\phi=\hat \phi_0-f_0$. It is smooth, it satisfies the Dirichlet condition, $\phi - \hat \phi_0$ belongs to $\mathscr S(\overline \Omega)$ and $\Delta\phi=1$. It remains to discuss the uniform positivity of the normal derivative. By the Hopf lemma we already know that $\partial_\nu \phi > 0$ on $\partial \Omega$, so it is enough to show that this estimate is uniform at infinity.

We have
\begin{equation*}%\label{eq.decomNnabla}
\partial_\nu \phi= \partial_\nu \hat \phi_0 - \partial_\nu f_0.
\end{equation*}
Since $\Theta$ is a rotation at infinity, we see by the explicit expression of $\phi_0$ that there exists $c_1>0$ such that, for all $x\in\partial\Omega$ with a sufficiently large curvilinear abscissa,
\begin{equation*}
\partial_\nu \hat \phi_0 \geq 2 c_1\,.
\end{equation*}
On the other hand, since $f_0 \in \mathscr S(\overline \Omega)$ we have 
\[\lim_{\substack{|x|\to+\infty\\ x \in \partial \Omega}} \partial_\nu f_0(x)=0\,.\]
Then $\partial_\nu \phi (x) \geq c_1$ for $x \in \partial\Omega$ large enough, and we deduce the uniform positivity of $\partial_\nu \phi$ on $\partial\Omega$.

\subsection{Proof of Proposition \ref{prop.assumptionvalid}}
For $(s,t) \in \Omega_0$ we set
\[
a(s,t) = \det \big( \mathrm{Jac} (\Theta)(s,t) \big) = 1 - t \kappa(s).
\]
Let $\tilde \phi = \phi \circ \Theta$. For $s \in \R$ and $\tau \in (-1,1)$ we set
\[
a_\delta(s,\tau) = a(s,\delta \tau) \quad \text{and} \quad \psi(s,\tau)=\delta^{-2} a_\delta(s,\tau)^{\frac 12} \tilde\phi(s,\delta\tau),
\]	
Finally we define on $\R \times (-1,1)$ the differential operator
\[
\mathscr{M}_\delta=\partial_\tau^2+\delta^2 \big( a^{-\frac12}_\delta\partial_sa_\delta^{-\frac12} \big)^2+\frac{\delta^2\kappa^2}{4 a_\delta^2}.
\]
\begin{lemma}\label{lem.psi}
We have $\mathscr{M}_\delta \psi=a_{\delta}^{\frac12}$ and $\psi( \cdot ,\pm1)=0$.
\end{lemma}

\begin{proof}
Since $\tilde\phi(s,\pm\delta)=0$ we have $\psi( \cdot ,\pm1)=0$ for all $s \in \R$.
In the tubular coordinates the equality $\Delta\phi=1$ reads
\[
\big( a^{-1}\partial_s a^{-1} \partial_s+ a^{-1}\partial_t a \partial_t \big) \tilde\phi=1\,.
\]
Setting $\check\phi =a^{\frac12}\tilde\phi$ we get
\[\left[\big(a^{-\frac12}\partial_sa^{-\frac12}\big)^2+\big(a^{-\frac12}\partial_ta^{\frac12}\big)\big(a^{\frac12}\partial_ta^{-\frac12}\big)\right]\check\phi=a^{\frac12}\,,\]
or
\[\left[\big(a^{-\frac12}\partial_sa^{-\frac12}\big)^2+ \Big(\partial_t-\frac{\kappa}{2a}\Big) \Big(\partial_t + \frac{\kappa}{2a}\Big)\right]\check\phi=a^{\frac12}\,,\]
which gives
\[\left[\big(a^{-\frac12}\partial_sa^{-\frac12}\big)^2+\partial_t^2+\frac{\kappa^2}{4 a^2}\right]\check\phi=a^{\frac12}\,.\]
Since $\psi(s,\tau)=\delta^{-2}\check\phi(s,\delta \tau)$, the conclusion follows.
\end{proof}

\begin{proof} [Proof of Proposition \ref{prop.assumptionvalid}]
We look for an approximation $\Psi_5$ of $\psi$, in the sense that 
\begin{equation} \label{eq.Psi-1}
\mathscr{M}_\delta(\psi-\Psi_5)=\mathscr{O}_{H^2(\R\times (-1,1))}(\delta^5)\,,\quad \psi-\Psi_5\in H^2\cap H^1_0(\R\times[-1,1])\,.
\end{equation}
By elliptic regularity this will give
\[
\|\psi-\Psi_5\|_{H^4(\R\times[-1,1])}=\mathscr{O}(\delta^3),
\]
and then, by Sobolev embeddings,
\begin{equation} \label{eq.Psi-2}
\|\psi-\Psi_5\|_{\mathscr{C}^2(\R\times[-1,1])}=\mathscr{O}(\delta^3)\,.
\end{equation}
We look for $\Psi_5$ of the form $\psi_0 + \delta \psi_1 + \delta^2 \psi_2 + \delta^3 \psi_3 + \delta^4 \psi_4$. Note that we could proceed similarly to get a rest of order $\mathscr O(\delta^N)$ in $\mathscr{C}^k(\R\times[-1,1])$ for any $N$ and $k$.

There exist $M_0,\dots,M_4 \in \mathcal L(H^{4}(\R\times[-1,1]),H^2(\R\times[-1,1]))$ such that in $\mathcal L(H^{4}(\R\times[-1,1]),H^2(\R\times[-1,1]))$ we have 
\[
\mathscr{M}_\delta = \sum_{k=0}^4 \delta^k M_k  + \mathscr O(\delta^5).
\]
In particular,
\[
M_0=\partial^2_\tau\,,\quad M_1=0\,,\quad M_2=\partial^2_s+\frac{\kappa^2}{4}\,.
\]
Similarly, in $H^2(\R\times[-1,1])$ we have by Lemma \ref{lem.psi}
\[
\mathscr M_\delta \psi = \sum_{k=0}^4 \delta^k\alpha_k+ \mathscr O (\delta^5),
\]
with
\[
\alpha_0  = 1, \quad \alpha_1 =  - \frac { \kappa \tau}2 , \quad  \alpha_2 =  - \frac { \tau^2 \kappa^2} 8 ,
\]
and $\alpha_3,\alpha_4 \in \mathscr C^\infty(\overline \Omega)$. We compute $\psi_k$ by induction on $k$. It satisfies
\[
M_0 \psi_k = - \sum_{j=2}^k M_j \psi_{k-j} + \alpha_k, \quad \psi_k (\cdot,\pm 1) = 0.
\]
This gives in particular
\[\psi_0(s,\tau)=\frac{\tau^2-1}{2}, \quad 
\psi_1(s,\tau)=\frac{\kappa(s)}{12}(\tau-\tau^3)\,.
\]
Then $\psi_2$ has to be a solution of
\[{M}_0 \psi_2=-{M}_2 \psi_0-\frac{\kappa^2\tau^2}{8}=\frac{\kappa^2}{4}\left(-\frac{\tau^2-1}{2}-\frac{\tau^2}{2}\right)=\frac{\kappa^2}{4}\left(\frac12-\tau^2\right)\,.\]
This leads to take
\[\psi_2(s,\tau)=\frac{\kappa^2}{4}\left(\frac{\tau^2-1}{4}-\frac{\tau^4-1}{12}\right)=\frac{\kappa^2}{4}\left(\frac{\tau^2}{4}-\frac{\tau^4}{12}-\frac16\right)\,.\]
Due to the asymptotic behavior of $\phi$ given in Proposition \ref{prop.phi}, $\psi-\psi_0$ belongs to the Schwartz class. Thus, $\Psi_5$ satisfies \eqref{eq.Psi-1} and hence \eqref{eq.Psi-2}. Now setting $\Psi = \psi_0 + \delta \psi_1 + \delta^2 \psi_2$ we deduce
\[
\|\psi-\Psi\|_{\mathscr{C}^2(\R\times[-1,1])}=\mathscr{O}(\delta^3)\,.
\]
This gives
\[
\|\delta^{-2}\tilde\phi(s,\delta\tau)-a(s,\delta\tau)^{-\frac12} \Psi\|_{\mathscr{C}^2(\R\times[-1,1])}=\mathscr{O}(\delta^3)
\]
or
\[
\left \| \delta^{-2}\tilde\phi(s,\delta\tau)-\left(1+\delta\tau\frac\kappa2+\delta^2\frac{3}{8}\tau^2\kappa^2\right) \Psi \right \|_{\mathscr{C}^2(\R\times[-1,1])}=\mathscr{O}(\delta^3)\,.
\]
Then
\begin{equation}\label{eq.approxC2}
\left\| \delta^{-2}\tilde\phi(s,\delta\tau)-f_\delta(s,\tau) \right\|_{\mathscr{C}^2(\R\times[-1,1])}
=\mathscr{O}(\delta^3)\,,
\end{equation}
where we have set 
\[
f_{\delta}(s,\tau)=\psi_0+\delta\left(\psi_1+\frac{\tau\kappa} 2\psi_0\right)+\delta^2\left(\frac{3\tau^2\kappa^2}{8}\psi_0+\frac{\tau\kappa}{2} \psi_1+ \psi_2\right)
\]
We have 
\[\begin{split}
f_{\delta}(s,\tau)
=\frac{\tau^2-1}{2}-\frac{\delta\kappa(s)}{6}\left(\tau-\tau^3\right)+\delta^2\kappa^2P_2(\tau)\,,
\end{split}\]
where 
\[P_2(\tau)=\frac{3\tau^2(\tau^2-1)}{16}+\frac{\tau(\tau-\tau^3)}{24}+\frac{\tau^2}{16}-\frac{\tau^4}{48}-\frac{1}{24}\,.\]

Let us explain why $f_\delta$ has a unique minimun, non attained at infinity, and which is non-degenerate.
Firstly, when $s\notin\mathrm{supp}\,\kappa$, we have 
\[
f_\delta(s,\tau)=\frac{\tau^2-1}{2}\geq-\frac12=f_\delta(s,0).
\]
This shows that $f_\delta$ has a minimum. 
This minimum is in fact strictly less than $-\frac12$ and thus attained at points where the curvature is not $0$. Indeed, consider $s_0$ the maximum of $\kappa^2$. We have $\kappa(s_0)\neq 0$, $\kappa'(s_0)=0$, and $\kappa(s_0)\kappa''(s_0)<0$.
Let us notice that
\begin{align*}
f_{\delta}\left(s_0,\frac{\delta\kappa(s_0)}{6}\right)
& =-\frac{1}{2}+\delta^2\kappa^2(s_0)\left(\frac{1}{72}-\frac{1}{36}-\frac{1}{24}\right)+\mathscr{O}(\delta^3)\\
& =-\frac{1}{2}-\frac{\delta^2\kappa^2(s_0)}{18}+\mathscr{O}(\delta^3)\,.
\end{align*}
This shows that, for $\delta$ small enough,
\[\inf_{(s,\tau)\in\R\times(-1,1)} f_{\delta}(s,\tau)\leq -\frac{1}{2}-\frac{\delta^2\max \kappa^2}{18}+C\delta^3<-\frac12\,,\]
and that the infimum is a minimum (which is not attained at infinity).

Now we prove that for $\delta$ small enough all the possible minima are non-degenerate.
Consider a minimum $(s_1,\tau_1)$ of $f_\delta$.  We have $\tau_1\in(-1,1)$ and $\kappa(s_1)\neq 0$. Moreover, we must have
\[\partial_\tau f_\delta(s_1,\tau_1)=0\,,\]
which implies that
\begin{equation}\label{eq.tau1critic}
\tau_1=\frac{\delta\kappa(s_1)}{6}+\mathscr{O}(\delta^2)\,.
\end{equation}
Then, 
\[f_\delta(s_1,\tau_1)=-\frac{1}{2}-\frac{\delta^2\kappa^2(s_1)}{18}+\mathscr{O}(\delta^3)\,.\]
With the upper bound on the minimum, we deduce that
\[0\leq \kappa^2(s_0)-\kappa^2(s_1)\leq C\delta\,.\]
By using the uniqueness and non-degeneracy of the minimum, this implies that
\begin{equation}\label{eq.approxmin}
s_1=s_0+\mathscr{O}(\delta^{\frac12})\,,\quad\tau_1=\frac{\delta\kappa(s_0)}{6}+\mathscr{O}(\delta^2)\,,
\end{equation}
where we used \eqref{eq.tau1critic} and that $\kappa'(s_0)=0$.

Let us now estimate the second derivative of $f_\delta$ at $(s_1,\tau_1)$. We have
\[\partial^2_s f_\delta(s_1,\tau_1)=-\frac{\kappa(s_0)\kappa''(s_0)}{9}\delta^2+o(\delta^2)\,,\quad \partial_s\partial_{\tau}f_\delta(s_1,\tau_1)=\mathscr{O}(\delta^{\frac32})\,,\]
and
\[\partial^2_\tau f_\delta(s_1,\tau_1)=1+\mathscr{O}(\delta^2)\,.\]
We infer that there exist $\delta_0,c>0$ such that for all $\delta\in(0,\delta_0)$ and all minimum $(s_1,\tau_1)$,
\[\mathrm{Hess}_{(s_1,\tau_1)}f_\delta\geq c\delta^2\,.\]
By definition, this means that the minima are non-degenerate.

Let us finally prove that there is only one minimum. Consider two minima $X_1=(s_1,\tau_1)$ and $X_2=(s_2,\tau_2)$. From \eqref{eq.approxmin}, we have, uniformly in $t\in[0,1]$, 
\begin{equation}\label{eq.X12}
X_1+t(X_2-X_1)=\left(s_0,\frac{\delta\kappa(s_0)}{6}\right)+(\mathscr{O}(\delta^{\frac12}),\mathscr{O}(\delta^2))\,.
\end{equation}
Since the differential of $f_\delta$ vanishes at $X_1$, the Taylor formula gives
\[f_\delta(X_2)-f_\delta(X_1)=\int_0^1(1-t) \mathrm{Hess}_{X_1+t(X_2-X_1)}f_\delta(X_2-X_1,X_2-X_1)\mathrm{d}t\,.\]
By using \eqref{eq.X12}, we deduce as before that there exist $\delta_0,c>0$ such that for all $\delta\in(0,\delta_0)$ and all $t\in[0,1]$,
\[\mathrm{Hess}_{X_1+t(X_2-X_1)}f_\delta\geq c\delta^2\,.\]
This shows that
\[0=f_\delta(X_2)-f_\delta(X_1)\geq \frac{\delta^2}{2}|X_1-X_2|^2\,.\]
Therefore, for $\delta$ small enough, $f_\delta$ has a unique minimum $X(\delta)$, which is not attained at infinity and non-degenerate, and 
\[\mathrm{Hess}_{X(\delta)}f_\delta\geq c\delta^2\,.\]
By a perturbative argument using \eqref{eq.approxC2}, this shows that $\delta^{-2}\tilde\phi(s,\delta\tau)$ has also a unique minimum, which is not attained at infinity and non-degenerate. The same conclusion follows for $\phi$.
\end{proof}

\section{Upper bound for the bottom of the spectrum}\label{sec.proof}
This last section is devoted to the proof of Theorem \ref{thm.main}. From the min-max principle, we have
\[\inf\mathrm{sp}(\mathscr{P}_h)=\inf_{\psi\in H^1_0(\Omega)\setminus\{0\}}\frac{\|(-ih\nabla-\mathbf{A})\psi\|^2-h\|\psi\|^2}{\|\psi\|^2}\,.\]
From Lemma \ref{lem.ephih}, we have
\begin{equation}\label{eq.infPh}
\inf\mathrm{sp}(\mathscr{P}_h)=\inf_{u\in H^1_0(\Omega)\setminus\{0\}}\frac{4h^2\int_{\Omega}e^{-2\phi/h}|\partial_{\overline{z}} u|^2\mathrm{d}x}{\int_{\Omega}e^{-2\phi/h}|u|^2\mathrm{d}x}\,.
\end{equation}
Let us construct a convenient test function. It is natural to consider a test function in the form
\[
u(x)=f(x)\chi(x)\,,
\]
where $f\in\mathscr{O}(\Omega)\cap H^1(\Omega)$ is such that $f(x_{\min})\neq 0$ and $\chi$ is of the form $\chi = \rho \circ \Theta^{-1}$ with $\rho(s,\pm\delta)=0$ and $\rho(s,t)= 1$ for all $s \in \R$ and $t\in(-\delta + \epsilon ,\delta- \epsilon)$. This function $\rho$ will be determined below to optimize an upper bound, see \eqref{eq.rho}.

\subsection{Estimate of the numerator}
By using the change of variable $\Theta$, we have
\begin{multline} \label{eq:dbarz-u}
4h^2\int_{\Omega}e^{-2\phi/h}|\partial_{\overline{z}} u|^2\mathrm{d}x \\
=h^2\int_{\Omega_0}e^{-2\tilde\phi(s,t)/h}|\tilde f(s,t)|^2\left(a^{-2}|\partial_s\rho|^2+|\partial_t\rho|^2\right) a(s,t)\mathrm{d}s\mathrm{d}t\,,
\end{multline}
with $\tilde\phi=\phi\circ\Theta$ and $\tilde f=f\circ\Theta$. Since $\rho$ is constant on $\R \times [-\delta + \epsilon, \delta - \epsilon]$, the right-hand side is actually an integral over $\R \times \big( (-\delta, -\delta + \epsilon) \cup (\delta-\epsilon , \delta) \big)$. We begin with the contribution of the integral over $\R \times (\delta-\epsilon , \delta)$. Thanks to the Taylor formula near $t=\delta$, we have
\begin{eqnarray*}
\lefteqn{\int_{\R}\mathrm{d}s\int_{\delta-\epsilon}^\delta \mathrm{d}t \,e^{-2\tilde\phi(s,t)/h}|\partial_t\rho|^2 |\tilde f(s,t)|^2a(s,t)}\\ &&\leq(1+C\epsilon+C\epsilon^2/h)\int_{\R}\mathrm{d}s\int_{\delta-\epsilon}^\delta \mathrm{d}t \,e^{-2(t-\delta)\partial_t\tilde\phi(s,\delta)/h}|\partial_t\rho|^2 |\tilde f(s,t)|^2a(s,\delta)\,.
\end{eqnarray*}
We also want to replace $ |\tilde f(s,t)|^2$ by $ |\tilde f(s,\delta)|^2$. To do so, we remark that, for all $(s,t)\in\R\times(\delta-\epsilon,\delta)$,
\begin{eqnarray*}
&\left||\tilde f(s,t)|^2-|\tilde f(s,\delta)|^2\right|&\leq 2\Re\int_{t}^\delta|\tilde f(s,\tau)||\partial_t\tilde f(s,\tau)|\mathrm{d}\tau\\
&&\leq\left(\|\tilde f(s,\cdot)\|^2_{L^2([\delta-\epsilon,\delta])}+\|\partial_t\tilde f(s,\cdot)\|^2_{L^2([\delta-\epsilon,\delta])}\right)\,,
\end{eqnarray*}
so that
\begin{multline*}
\int_{\R}\mathrm{d}s\int_{\delta-\epsilon}^\delta \mathrm{d}t \,e^{-2(t-\delta)\partial_t\tilde\phi(s,\delta)/h}|\partial_t\rho|^2a(s,\delta)\left| |\tilde f(s,t)|^2- |\tilde f(s,\delta)|^2\right|\\
\leq \int_{\mathbb{R}}\mathrm{d}s\, a(s,\delta) R(s,\epsilon,h)\,,
\end{multline*}
with
\begin{multline}\label{eq.Rseh}
R(s,\epsilon,h)
=\left(
\|\tilde f(s,\cdot)\|^2_{L^2([\delta-\epsilon,\delta])}
+\|\partial_t\tilde f(s,\cdot)\|^2_{L^2([\delta-\epsilon,\delta])}
\right)\int_{\delta-\epsilon}^\delta \mathrm{d}t \,e^{-2(t-\delta)\partial_t\tilde\phi(s,\delta)/h}|\partial_t\rho|^2\,.
\end{multline}
Therefore,
\begin{equation}\label{eq.ubpartialt}
	\begin{split}
	\lefteqn{\int_{\R}\mathrm{d}s\int_{\delta-\epsilon}^\delta \mathrm{d}t \,e^{-2\tilde\phi(s,t)/h}|\partial_t\rho|^2 |\tilde f(s,t)|^2a(s,t)}\\ \leq&(1+C\epsilon+C\epsilon^2/h)\int_{\R}\mathrm{d}s\int_{\delta-\epsilon}^\delta \mathrm{d}t \,e^{-2(t-\delta)\partial_t\tilde\phi(s,\delta)/h}|\partial_t\rho|^2 |\tilde f(s,t)|^2a(s,\delta)\\
	\leq& (1+C\epsilon+C\epsilon^2/h)\Big(\int_{\R}\mathrm{d}s\int_{\delta-\epsilon}^\delta \mathrm{d}t \,e^{-2(t-\delta)\partial_t\tilde\phi(s,\delta)/h}|\partial_t\rho|^2 |\tilde f(s,\delta)|^2a(s,\delta)\\
	&\qquad\qquad\qquad\qquad\qquad\qquad\qquad\qquad\qquad+\int_{\mathbb{R}}\mathrm{d}s\, a(s,\delta) R(s,\epsilon,h)\Big)\,.
	\end{split}
\end{equation}

Looking at the right-hand-side suggests to consider a function $\rho$ that  minimizes $\int_{\delta-\epsilon}^\delta \mathrm{d}t \,e^{-2(t-\delta)\partial_t\tilde\phi(s,\delta)/h}|\partial_t\rho|^2$ among the $H^1$-functions equal to $1$ in $\delta-\epsilon$ and $0$ in $\delta$. This leads to the explicit choice
\begin{equation}\label{eq.rho}
	\rho(s,t)= \frac{1 - e^{2(t-\delta)\partial_t\tilde\phi(s,\delta)/h}}{1 - e^{-2\epsilon\partial_t\tilde\phi(s,\delta)/h}}\,,\quad\forall (s,t)\in\R\times(\delta-\epsilon,\delta)\,.
\end{equation}
The minimum satisfies
\[
\int_{\delta-\epsilon}^\delta \mathrm{d}t \,e^{-2(t-\delta)\partial_t\tilde\phi(s,\delta)/h}|\partial_t\rho|^2
= \frac{2\partial_t\tilde\phi(s,\delta)}{h(1-e^{-2\varepsilon\partial_t\tilde\phi(s,\delta)/h})}\,.
\]
We recall from Proposition \ref{prop.phi} that $\partial_t\tilde\phi(s,\delta) = \partial_\nu \phi (\Theta(s,\delta))$ is uniformly positive.
Choosing $\epsilon=h|\ln h|$, we get, uniformly with respect to $s$,
\begin{equation}\label{eq.dphi/h}
\int_{\delta-\epsilon}^\delta \mathrm{d}t \,e^{-2(t-\delta)\partial_t\tilde\phi(s,\delta)/h}|\partial_t\rho|^2
= \frac{2\partial_t\tilde\phi(s,\delta)}{h}+o(h^{-1})=\mathscr{O}(h^{-1})\,,
\end{equation}
where we used that $\Theta$ and  $\Theta^{-1}$ have uniformly bounded Jacobians. 

Using that $f\in H^1(\Omega)$, we get
\[
\int_{\mathbb{R}}\mathrm{d}s\left(
\|\tilde f(s,\cdot)\|^2_{L^2([\delta-\epsilon,\delta])}
+\|\partial_t\tilde f(s,\cdot)\|^2_{L^2([\delta-\epsilon,\delta])}
\right)\underset{\epsilon\to0}{\longrightarrow}0\,,
\]
so that, with \eqref{eq.Rseh} and \eqref{eq.dphi/h}, it follows that
\[\int_{\mathbb{R}}\mathrm{d}s\, a(s,\delta) R(s,\epsilon,h)=o_{h\to 0}(h^{-1})\,.\]
With \eqref{eq.ubpartialt}, this gives
\begin{eqnarray*}
\lefteqn{\int_{\R}\mathrm{d}s\int_{\delta-\epsilon}^\delta \mathrm{d}t \,e^{-2\tilde\phi(s,t)/h}|\partial_t\rho|^2 |\tilde f(s,t)|^2a(s,t)}
\\ &&
\leq 2h^{-1}\int_{\R}\partial_\nu \phi (\Theta(s,\delta))|\tilde f(s,\delta)|^2a(s,\delta)\mathrm{d}s +o_{h\to0}(h^{-1})\,.
\end{eqnarray*}
Let us now come back to \eqref{eq:dbarz-u}. Considering the term with the tangential derivative, we get with similar computations
\begin{eqnarray*}
\int_{\R}\mathrm{d}s\int_{\delta-\epsilon}^\delta \mathrm{d}t \,e^{-2\tilde\phi(s,t)/h}|\partial_s\rho|^2 |\tilde f(s,t)|^2a(s,t)
=o_{h\to0}(h^{-1})\,.
\end{eqnarray*}
We play the same game with the contribution of the integral over $\R \times (-\delta,-\delta+\epsilon)$ in \eqref{eq:dbarz-u} (notice that $\partial_t\tilde\phi(s,-\delta) = -\partial_\nu \phi (\Theta(s,-\delta))$ is now uniformly negative). 
We get
\begin{align*}
4h^2\int_{\Omega}e^{-2\phi/h}|\partial_{\overline{z}} u|^2\mathrm{d}x
& \leq  2h\|(\partial_\nu \phi)^{\frac12}f\|^2_{\partial\Omega}+o_{h\to0}(h)\,.
\end{align*}
\subsection{Estimate of the denominator and conclusion}
We have
\begin{align*}
\int_{\Omega}e^{-2\phi/h}|u|^2\mathrm{d}x
& = \int_{\Omega}e^{-2\phi/h}|f(x)\chi(x)|^2\mathrm{d}x \\
& = e^{-2\phi_{\min}/h}\int_{\Omega}e^{-2(\phi-\phi_{\min})/h}|f(x)\chi(x)|^2\mathrm{d}x\,.
\end{align*}
The Laplace method yields
\[\int_{\Omega}e^{-2\phi/h}|u|^2\mathrm{d}x= h e^{-2\phi_{\min}/h}\left( |f(x_{\min})|^2\frac{\pi}{\sqrt{\det\mathrm{Hess}_{x_{\min}}\phi}}+o_{h\to0}(1)\right)\,.\]
With \eqref{eq.infPh}, this shows that
\[\inf\mathrm{sp}(\mathscr{P}_h)\leq 2\sqrt{\det\mathrm{Hess}_{x_{\min}}\phi}\frac{\|(\partial_\nu \phi)^{\frac12}f\|^2_{\partial\Omega}}{\pi|f(x_{\min})|^2}(1+o_{h\to0}(1))e^{2\phi_{\min}/h}\,,\]
and Theorem \ref{thm.main} since this estimate holds for all the functions $f$ in $\mathscr{E}$.

\section*{Acknowledgments}
N.R. is grateful to David Krej{\v{c}}i{\v{r}}{\'{\i}}k for sharing the reference \cite{Exner12}.

\bibliographystyle{abbrv}
\bibliography{biblio}

\end{document}